\documentclass[conference]{IEEEtran}
\IEEEoverridecommandlockouts
\usepackage{cite}
\usepackage{amsmath,amssymb,amsfonts}
\usepackage{algorithmic}
\usepackage{graphicx}
\usepackage{textcomp}
\usepackage{xcolor}
\def\BibTeX{{\rm B\kern-.05em{\sc i\kern-.025em b}\kern-.08em
    T\kern-.1667em\lower.7ex\hbox{E}\kern-.125emX}}

    \usepackage{amsfonts}     
\usepackage{graphicx}
\usepackage{color}
\usepackage{pgf, tikz, pgfplots}
\usepackage{siunitx}
\pgfplotsset{compat=1.17} 
\usepackage{amssymb} 
\usepackage{bm}
\usepackage{amsthm}
\usepackage{acronym}
\usepackage{tikz}
\usepackage{mathrsfs}
\usepackage{enumerate}
\usetikzlibrary{arrows,calc,fit,matrix,positioning,shapes,shadows,trees,mindmap,tikzmark,arrows.meta,angles,quotes,babel,spy}

\usepackage{booktabs}    

\newtheorem{theorem}{Theorem}

\DeclareMathOperator*{\mathboxplus}{\boxplus}
\newcommand{\PermAut}{$\mathrm{Aut}(\mathcal{C})$}
\newcommand{\Aut}{$\mathrm{GAut}(\mathcal{C})$}
\DeclareMathOperator*{\mPermAut}{\mathrm{Aut}(\mathcal{C})}

\begin{document}

\title{Generalized
Automorphisms of Channel Codes: Properties, Code Design, and a Decoder\\
\thanks{The author would like to thank Frank Herrlich for the inspiring discussions in deriving Theorem \ref{theorem:core}.

This work has received funding from the 
German Federal Ministry of Education and Research (BMBF) within the project Open6GHub (grant agreement 16KISK010).}
}

\author{\IEEEauthorblockN{Jonathan Mandelbaum, Holger Jäkel, and Laurent Schmalen}
\IEEEauthorblockA{Communications Engineering Lab, Karlsruhe Institute of Technology (KIT), 76131 Karlsruhe, Germany\\
\texttt{jonathan.mandelbaum@kit.edu}}
}

\maketitle

\begin{abstract}
Low\--density parity\--check codes together with belief propagation (BP) decoding are known to be well\--performing
for large block lengths. However, for short block lengths there is still a considerable gap between the performance of BP decoding and
maximum likelihood decoding. Different ensemble decoding schemes such as, e.g., automorphism ensemble decoding (AED), can reduce this gap in short block length regime.
We propose generalized AED (GAED) that uses automorphisms according to the definition in linear algebra.
Here, an automorphism of a vector space is defined as a linear, bijective self-mapping, whereas in coding theory self-mappings that are scaled 
permutations are commonly used.
We show that the more general definition leads to an explicit joint construction of codes and automorphisms, and significantly enlarges the search space for automorphisms of existing linear codes.
Furthermore, we prove the concept that generalized automorphisms can indeed be used to improve decoding.
Additionally, we propose a code construction of linear codes enabling
the construction of codes with suitably designed automorphisms. Finally, we analyze the decoding performances of GAED for some of our constructed codes.
\end{abstract}


\section{Introduction}
\label{sec:introduction}

Low\--density parity\--check (LDPC) codes are a prominent example of error correcting codes that are used
in a large variety of applications.
 They were first proposed by Gallager together with a low\--complexity message passing decoding algorithm \cite{Gallager_LDPC_diss}, often called belief propagation
 (BP) decoding.
 LDPC codes with large block lengths can achieve low error rates and close\--to\--capacity performance\cite{MCT08}.
 Yet, many low\--latency communication systems,
 such as the internet of things, autonomous driving, or communicating control commands, require codes of short block lengths. %
 For such codes, there is still a considerable gap between the performance of BP decoding 
 and maximum likelihood (ML) decoding.
 This can be attributed to the poor structural properties
 of the parity\--check matrix (PCM) of short codes for BP decoding, e.g., a large number of short cycles and of non\--zero entries, together with the sub\--optimality of the message passing decoding algorithm.

 For any binary\--input memoryless symmetric output channel (BMSC), such as the additive white Gaussian noise (AWGN) channel, correctly decoding a received word with a symmetric message passing decoding algorithm depends only on the noise superimposed by the channel, not on the transmitted codeword itself \cite[Lemma 4.90]{MCT08}.
Therefore, different variations of the algorithm such as
 multiple bases belief propagation (MBBP) \cite{MBBP2} and automorphism ensemble decoding (AED) \cite{stuttgart_ldpc_aed} were proposed.
Herein, an ensemble of noise representations or decoding algorithms 
aims to improve decoding.
 Following the latter statement, MBBP decoding performs ensemble decoding, in which the received sequence is decoded 
 by multiple different BP decoders in parallel.
  Similarly, the idea of the AED, as introduced in \cite{AED_RMcodes},
is to use the automorphism group defined in \cite{MacWilliamsSloane} to exploit different noise representations in parallel paths which is discussed in more detail in Sec. \ref{sec:aed}.

From a structural perspective, it is reasonable
 to restrict the definition of the automorphism groups of linear codes to consist
 solely of suitable \emph{scaled permutations} (Sec. \ref{subsec:aut}) such that, i.a., codewords of the same Hamming weight are mapped onto each other.
 The (permutation) automorphism groups of classical and modern codes are extensively discussed in the literature \cite{stuttgart_ldpc_aed}, \cite{MacWilliamsSloane},\cite{BCH_Aut}, 
 \cite{Aut_PolarCodes_Geiselhart}.  
To improve decoding, automorphisms must be chosen carefully because the generated diversity might be absorbed
by the symmetry of the decoding algorithm \cite{stuttgart_ldpc_aed},\cite{AED_RMcodes},\cite{EnhancingCyclic}.
In linear algebra, the definition of automorphisms of a vector space is broader and includes all linear, bijective self\--mappings, a fact that is going to be used and analyzed in this paper.

 We show that automorphisms according to this
 more general definition can be beneficial for decoding. Thus, we significantly enlarge the search space for suitable automorphisms.
 To do so, we describe the (generalized) automorphism group of linear codes defined by a PCM. Additionally, we propose a generalized AED (GAED) algorithm such that generalized automorphisms can be used for decoding.  
Furthermore, we propose a code construction algorithm for linear codes together with specific automorphisms which enables designing suitable codes for GAED.
Finally, we present and compare the performance of GAED for some of our constructed codes. 

\tikzset{myblock/.style={rectangle, draw, thin, minimum width=0.6cm, minimum height=0.6cm},font=\footnotesize,align=center}%
\tikzset{mywideblock/.style={rectangle, draw, thin, minimum width=1cm, minimum height=0.6cm},font=\footnotesize,align=center}%

\newcommand{\myline}[2]{
\path(#1.east) --(#2.west)  coordinate[pos=0.4](mid);
\draw[-latex] (#1.east) -| (mid) |- (#2.west);
}
\section{Preliminaries}%
\label{sec:Preliminaries}
A linear block code $\mathcal{C}(n,k)$ over a field~$\mathbb{F}$ forms a subspace of the vector space $\mathbb{F}^n$.
It consists of $|\mathbb{F}|^k$ distinct elements from $\mathbb{F}^n$,
where the parameters $n\in \mathbb{N}$ and $k\in \mathbb{N}$ are called block length and information length, respectively.
A linear code $\mathcal{C}(n,k)$ can be described as the row span of a generator matrix $\bm{G}\in \mathbb{F}^{k\times n}$ or as the null space of its PCM $\bm{H}\in \mathbb{F}^{(n-k)\times n}$, which we assume to be of full rank \cite{MacWilliamsSloane}:
$$\mathcal{C}\left(n,k\right)=\left\{\bm{x}\in \mathbb{F}^n:\bm{H} \bm{x} = \bm{0}\right\}=\mathrm{Null}(\bm{H}).$$
Note that in contrast to most coding literature, we denote vectors as column vectors
in order to directly account for matrix-vector operations common in linear algebra.

For the sake of simplicity, the parameters of the code will be $(n,k)$ and omitted if they are clear from the context.

 BP decoding is an iterative message passing algorithm over the Tanner graph of the code. 
Messages are log\--likelihood ratios (LLRs) that are iteratively propagated along the edges and updated in the nodes of the graph \cite{MCT08}.
 Every linear code can be represented by possibly different PCMs or, equivalently, different Tanner graphs. Although the code is the same, BP decoding behaves differently since the degrees of the nodes and the short cycles within the Tanner graph mainly dominate their performance.
For more details on BP decoding, the interested reader is referred to \cite{MCT08}.

\vspace{-1mm}

\subsection{Automorphism Group}\label{subsec:aut}
In this section, we discuss two different definitions of the automorphism group of a code. To this end, let $\mathcal{C}\subset \mathbb{F}^n$ be a linear code defined over an arbitrary finite
field $\mathbb{F}$ and $\mathrm{S}_n$ be the symmetric group.

First, according to \cite{MacWilliamsSloane}, the automorphism group is defined as the set of mappings $\pi^{(a)}$ %
with
\begin{equation*}
    \mathrm{Aut}(\mathcal{C}):=  \!  
    \left\{\! \pi^{(a)}: \mathcal{C}\!\to\!\mathcal{C}, \bm{x}\!\mapsto\!a\pi(\bm{x})\!:\! \pi\!\in \mathrm{S}_n, a\in\mathbb{F}\!\setminus\!\{0\}    \right\},
\end{equation*}
where $a{\pi(\bm{x})=\begin{pmatrix}
    ax_{\pi(1)},&\cdots &,ax_{\pi(n)}
\end{pmatrix}}^\mathsf{T}$ can be interpreted as a \emph{scaled permutation}.
Note that for binary codes, the scaling factor $a$ must be $1$ and, hence, is omitted, i.e., $\pi:=\pi^{(1)}$.

Second, in linear algebra another definition is standard. Here, 
the automorphism group $\mathrm{GAut}(\mathcal{C})$ of a vector space $\mathcal{C}$ is defined as all
    linear, bijective self-mappings \cite{bhattacharya1994basic}, i.e., \vspace{-1mm}
    \begin{equation*}
        \mathrm{GAut}(\mathcal{C}):=    
        \left\{ \tau: \mathcal{C}\rightarrow\mathcal{C}: \text{ $\tau$ linear, $\tau$ bijective} \right\}.\vspace{-2mm}
\end{equation*}
Since scaled permutations are linear, bijective mappings and, thus, $\mathrm{Aut}(\mathcal{C})\subseteq \mathrm{GAut}(\mathcal{C})$, the latter definition is more general.

\vspace*{-1mm}
\subsection{Automorphism Ensemble Decoding}
\label{sec:aed}
Fig. \ref{figure:aed} shows the block diagram of AED as proposed in \cite{AED_RMcodes} for binary codes. 
A codeword ${\bm{x}\in \mathcal{C}\subset \mathbb{F}_2^n}$ is transmitted over a BMSC yielding a received word $\bm{y}\in \mathcal{Y}^n$, where $\mathcal{Y}$ denotes the channel output alphabet, and resulting in the
bit-wise LLR vector  ${\bm{L}:=\left(L(y_j|x_j)\right)_{j=1}^n\in \mathbb{R}^n}$.
Instead of decoding the LLR vector  $\bm{L}$ with only one decoder, it is propagated along $K$ different paths. In each path ${i\in \{1,\dots, K\}}$, the LLR vector  $\bm{L}$ is preprocessed
according to a permutation automorphism ${\pi_i\in \mathrm{Aut}(\mathcal{C})}$ as $\pi_i(\bm{L})$.
Afterward, the permuted LLRs $\pi_i(\bm{L})$ are decoded using an arbitrary decoding algorithm of the code $\mathcal{C}$.
This yields several estimates of the permuted versions of the codeword
 $\pi_i(\hat{\bm{x}}_i)$.
Hereby, every path might possess a distinct decoder.
In the next step, applying the inverse automorphisms results in $K$
 estimates ${\hat{\bm{x}}_i\in\mathbb{F}_2^n}$ of the transmitted codeword ${\bm{x}\in\mathbb{F}_2^n}$. Finally,
 the best candidate is chosen according to an \emph{ML\--in\--the\--list} rule \cite{AED_RMcodes}.

\begin{figure}[t]
    \centering
\begin{tikzpicture}	  

    \node (input) at (-1.5,-0.5) {$\bm{y}$};    
    \draw[fill] (-0.975,-0.5) circle (1pt);
    \node (pi1) [mywideblock] at (0,0.3) {$\pi_1$};
    \node (pik) [mywideblock] at (0,-1.3) {$\pi_{K}$};

    \node (debox_1) [myblock, right=0.8cm of pi1] {Dec.};
    \node (debox_K) [myblock, right=0.8cm of pik] {Dec.};

    \node (pi_inv_1) [mywideblock,right=0.8cm of debox_1]  {$\pi_1^{-1}$};
    \node (pi_inv_k) [mywideblock,right=0.75cm of debox_K]  {$\pi_{K}^{-1}$};

    \node (ml_1) [right=0.8cm of pi_inv_1] {};
    \node (ml_K) [right=0.84cm of pi_inv_k] {};

    \node (ml_box) [rectangle, draw, thin, minimum width=0.6cm, minimum height=2.2cm] at (5,-0.5)  {\rotatebox{90}{ML-in-the-list}};
    \node (output) [right=0.3cm of ml_box] {$\hat{\bm{x}}$};
    
    \node (dots_box) [below=0cm of pi1]  {$\vdots$};
    \node (dots_box) [below=0cm of debox_1]  {$\vdots$};
    \node (dots_pi_inv_box) [below=0cm of pi_inv_1] {$\vdots$};
    \myline{input}{pi1};
    \myline{input}{pik};

    \draw [-latex] (pi1) -- (debox_1);
    \draw [-latex] (pik) -- (debox_K);

    \draw [-latex] (debox_1) -- (pi_inv_1);
    \draw [-latex] (debox_K) -- (pi_inv_k);

    \draw [-latex] (pi_inv_1.east) -- (ml_1);
    \draw [-latex] (pi_inv_k) -- (ml_K);	

    \draw [-latex] (ml_box) -- (output);	

\end{tikzpicture}
\caption{Block diagram of an AED. $K$ different automorphism ${\pi_i \in \mPermAut}$ are chosen from the permutation automorphism group \cite{AED_RMcodes}.}\label{figure:aed}
\vspace{-0.5cm}
\end{figure}
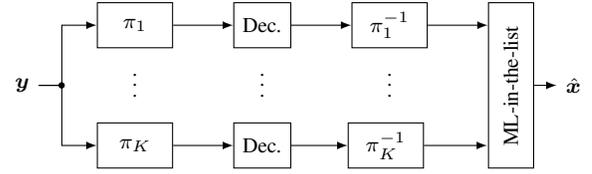

AED relies on the assumption that if decoding fails in one path, it may succeed in another path. This is not always
the case. First, it depends on the interaction of the respective automorphism and the decoder in a path. Second, 
the ensemble of automorphisms for the different paths must be chosen carefully to improve decoding performance \cite{Aut_PolarCodes_Geiselhart}.

We consider linear codes that are decoded with a BP decoder using a flooding schedule.
As discussed in \cite{stuttgart_ldpc_aed} and \cite{EnhancingCyclic}, the known 
automorphisms from the permutation automorphism group 
cannot be used to improve BP decoding for a 
large variety of LDPC codes since their diversity is absorbed 
by the symmetry of the PCM resulting from code construction.
Thus, either the PCM must be altered after construction as in \cite{stuttgart_ldpc_aed} or other construction 
methods must be used as in \cite{EnhancingCyclic}.

\vspace{-1mm}

\subsection{Frobenius Normal Form}
\label{sec:frobenius}
The code construction presented below in Sec. \ref{sec:code-construction} is based on the Frobenius normal form. To this end, let $\bm{T}\in \mathrm{GL}_n(\mathbb{F})$, with $\mathrm{GL}_n(\mathbb{F})$ denoting the general linear group, be a non-singular matrix describing a linear, bijective self\--mapping $\tau:\mathbb{F}^n\rightarrow \mathbb{F}^n$ via $\tau(\bm{x}):= \bm{T} \bm{x}$.
Then, there exists a matrix ${\bm{F}\in \mathrm{GL}_n(\mathbb{F})}$
  of the form 
  \begin{equation*}
    \bm{F}=\begin{pmatrix}
      \bm{B}_{f_1}& 0& \cdots&0 \\
      0 & \ddots & \ddots &\vdots \\
      \vdots & \ddots & \ddots  & 0 \\
      0 & \cdots & 0 & \bm{B}_{f_{j}}
    \end{pmatrix}
  \end{equation*}
  with ${j\in \mathbb{N}}$, consisting of companion matrices
 $\bm{B}_{f_i}$ \cite[p. 106]{MacWilliamsSloane} of size $d_i\times d_i$ of a polynomial $f_i(x)=\sum_{\ell=0}^{d_i} \alpha_{i, \ell} x^\ell$ with $ \alpha_{i, \ell}\in \mathbb{F}$,
together with a matrix $\bm{S}_{\mathrm{F}}\in \mathrm{GL}_n(\mathbb{F})$ such that 
$
  \bm{T}=\bm{S}_{\mathrm{F}}^{-1} \bm{F}  \bm{S}_{\mathrm{F}}$ \cite{bhattacharya1994basic}. 
The matrix $\bm{F}$ is called the \emph{Frobenius normal form} or, equivalently, \emph{rational canonical form} of $\bm{T}$ and
 is determined by $\bm{T}$ except for the order of the 
blocks~$\bm{B}_{f_{i}}$.

\section{Generalized Automorphisms of Codes}
This section investigates the general automorphism group \Aut{} of linear codes. Afterward, we propose code design and an adaption of AED that take advantage of using general automorphisms from $\text{\Aut{}}$.

\subsection{Automorphism Group of Parity\--Check Codes}
\label{sec:main-sec}
In the following, we will derive algebraic properties providing automorphisms of arbitrary codes and enabling a joint construction of linear codes together with their automorphisms. This will be approached by finding automorphisms of $\mathbb{F}^n$ and restricting them to the code subspace ${\mathcal{C}\subset\mathbb{F}^n}$. Directly finding automorphisms of $\mathcal{C}$ and exploiting the associated flexibility is part of our ongoing research.

A linear, bijective self\--mapping ${\tau:\mathbb{F}^n\rightarrow\mathbb{F}^n}$ based on a non-singular transformation matrix $\bm{T}\in \mathrm{GL}_n(\mathbb{F})$ is an automorphism of a code $\mathcal{C}$ if and only if
$$\forall \bm{x}\in \mathcal{C}:\quad \bm{H} \bm{x}=\bm{0}\Longleftrightarrow \bm{H} \bm{T} \bm{x}=\bm{0}.$$
Thus, in order to identify the automorphisms of a code, non-singular matrices can be used that retain the null space 
of $\bm{H}$ under right multiplication, i.e.,
\begin{equation}\label{equation:main_theorem}
  \mathrm{Null}(\bm{H})=\mathrm{Null}(\bm{H} \bm{T}),  
\end{equation}
which is investigated below in Theorem \ref{theorem:core}. 
In order to state the theorem, the following definitions are useful:  
\begin{align*}
    \mathcal{T}
    &:=
    \left\{
    \bm{T} \in \mathrm{GL}_n(\mathbb{F}): \bm{T} \text{ fulfills (\ref{equation:main_theorem}) } 
    \right\},
    \\
    \mathcal{Z}(n,k)
    &:=
    \left\{\bm{Z} \in \mathrm{GL}_n(\mathbb{F}):
    \bm{Z}=
     \begin{bmatrix}  \bm{C} &\bm{0}_{(n-k)\times k} \\
     \bm{D}& \bm{E}\end{bmatrix}  \right\}
     ,
\end{align*}
with $\bm{C}\in\mathbb{F}^{(n-k)\times (n-k)}$, $\bm{D}\in\mathbb{F}^{k\times (n-k)}$, and $\bm{E}\in\mathbb{F}^{k\times k}$. As before, for the sake of simplicity, the parameters $(n,k)$ will be omitted if they are clear from the context.

\begin{theorem}\label{theorem:core} 
Let $\bm{H}\in \mathbb{F}^{(n-k)\times n}$ be of rank $n-k$. 
     Then, $\mathcal{T}$ forms a subgroup
     of the general linear group $\mathrm{GL}_n(\mathbb{F})$ which is conjugated to the matrices in $\mathcal{Z}$, i.e.: 
 $\bm{T}\in\mathcal{T}$ if and only if there exists $\bm{Z}\in\mathcal{Z}$ and $\bm{A}\in\mathrm{GL}_n(\mathbb{F})$ such that $\bm{T}=\bm{A} \bm{Z}  \bm{A}^{-1}$,
where $\bm{A}\in \mathrm{GL}_n(\mathbb{F})$ is a non-singular matrix, termed
\emph{code characterization matrix (CCM)}, such that
\begin{equation}\label{eq:htohtilde}
\bm{H} \bm{A}=\tilde{\bm{H}}=\begin{bmatrix}
    \bm{I}_{(n-k)\times (n-k)}& \bm{0}_{(n-k)\times k}
\end{bmatrix}.
\end{equation}
\end{theorem}

\begin{proof}
First, we consider the case ${\bm{H}=\tilde{\bm{H}}}$.
The null space of $\tilde{\bm{H}}$ is the linear span of the $k$ canonical vectors ${\bm{e}_{n-k+1},\dots ,\bm{e}_{n}\in \mathbb{F}^n}$ where the $i^{\text{th}}$ entry of $\bm{e}_i$ is $1$ and all other entries are $0$.
Then, the subspace $\mathrm{Null}(\tilde{\bm{H}})$ is mapped on itself by all matrices $\bm{Z}\in \mathcal{Z}$, i.e.,
 \begin{equation}
    \mathrm{span}\{\bm{Z} \bm{e}_{n-k+1},\dots ,\bm{Z} \bm{e}_{n}\}=\mathrm{span}\{\bm{e}_{n-k+1},\dots, \bm{e}_{n}\},
 \end{equation}
 proving the theorem in the case of $\bm{H}=\tilde{\bm{H}}$.
 
An arbitrary matrix $\bm{H}\in\mathbb{F}^{(n-k)\times n}$ of rank $n-k$ can be transformed into $\tilde{\bm{H}}$ by applying Gaussian elimination on the columns which can be represented by multiplication of $\bm{H}$ from the right with a non-singular matrix $\bm{A}$ as in (\ref{eq:htohtilde}).

Consider again the basis $\{\bm{e}_{n-k+1},\dots, \bm{e}_{n}\}$ of $\mathrm{Null}(\tilde{\bm{H}})$.
 Then, using (\ref{eq:htohtilde}) and arbitrary $i \in \{1,\dots,k\}$, it follows that
$%
     \bm{0}=\tilde{\bm{H}} \bm{e}_{n-k+i} = \bm{H} \bm{A}  \bm{e}_{n-k+i}.
$%
 Thus, $\bm{A}  \bm{e}_{n-k+i}$ is an element of $\mathrm{Null}(\bm{H})$.
 In addition, since $\bm{A}$ is non-singular and the vectors in $\{\bm{e}_{n-k+1}, \ldots, \bm{e}_n\}$ are linearly independent, it follows that the vectors in $\{\bm{A}  \bm{e}_{n-k+1}, \ldots,\bm{A}  \bm{e}_{n}\}$
are also linearly independent. 
Therefore, $\{\bm{A}  \bm{e}_{n-k+1},\ldots, \bm{A}  \bm{e}_{n}\}$ is a basis of $\mathrm{Null}(\bm{H})$. 

The CCM $\bm{A}$ is a non\--unique change\--of\--basis matrix
mapping the basis $\{\bm{e}_{n-k+i}\}_{i=1}^{k}$ of $\mathrm{Null}(\tilde{\bm{H}})$ to the basis $\{\bm{A}  \bm{e}_{n-k+i}\}_{i=1}^{k}$ of $\mathrm{Null}(\bm{H})$.
Hence, $\bm{A}^{-1}$ is a change\--of\--basis matrix from the set $\{\bm{A}  \bm{e}_{n-k+i}\}_{i=1}^{k}$ to the set $\{\bm{e}_{n-k+i}\}_{i=1}^{k}$.
 Thus, conjugation of arbitrary $\bm{Z}\in \mathcal{Z}$ with $\bm{A}$ leads to all matrices that fulfill (\ref{equation:main_theorem}).
\end{proof}

 Theorem \ref{theorem:core} provides some structural insights into the automorphism group of a code and states an explicit construction method of automorphisms.
 Another relevant property of its CCM for the code design is highlighted in Theorem \ref{lemma:codeinfo_in_a}  and proven in the appendix.%

  \begin{theorem}\label{lemma:codeinfo_in_a} All characteristics of a linear code with PCM~${\bm{H}\in \mathbb{F}^{(n-k)\times n}}$, except for its code rate, are contained within the non-singular CCM $\bm{A}\in \mathrm{GL}_n(\mathbb{F})$.
    In addition, the inverse of $\bm{A}$ is of the form:
    \begin{equation} \vspace{-1mm}\label{equation:ainv_code_info}
\bm{A}^{-1}=\begin{pmatrix}
    \bm{H}\\
    \bm{\Lambda}
\end{pmatrix},\vspace{-1mm}
    \end{equation}
    where $\bm{\Lambda}\in \mathbb{F}^{k\times n}$ must be chosen such that $\bm{A}^{-1}$ is of full rank.
    Furthermore, the CCM $\bm{A}$ is not unique.
 \end{theorem}
 
\subsection{Generalized Automorphism Ensemble Decoding}
\label{sec:preprocessing}
The statements in Sec. \ref{sec:main-sec} hold for arbitrary fields. 
In this section, we confine ourselves to binary codes $\mathcal{C}\subset\mathbb{F}_2^n$
to adapt AED for automorphisms of ${\text{\Aut{}}\setminus \text{\PermAut{}}}$. 
We require a preprocessing of the bit\--wise LLRs, as depicted in Fig. \ref{figure:aut_preprocessing}, for using generalized automorphisms in GAED.
In order to describe the effects of $\mathbb{F}_2$-sums on the LLRs,
let $(X_i)_{i=1}^{s}\in \mathbb{F}_2^s$ be a sequence of
binary
random variables. Accordingly, let ${(L_i)_{i=1}^{s} \in \mathbb{R}^s}$ be their corresponding LLRs. 
 Then, the LLR of the $\mathbb{F}_2$-sum is given by the box-plus operator \cite{485714}: %
  \begin{equation}\label{eq:boxplus}
      L\left(\sum\limits_{i=1}^{s} X_i\right)=2\cdot \tanh^{-1}\left(\prod_{
 i=1}^{s} \tanh \left(\frac{L_i}{2}\right) \right) =: \mathboxplus_{i=1}^{s} L_i.
  \end{equation}

 \begin{figure}[t]
    \centering
    \begin{tikzpicture}	[level 1/.style={sibling distance=20mm},edge from parent/.style={->,draw},>=latex]
		\node (in) at (0,0) {$\bm{y}$};
		\node (prepro) [myblock,right=0.7cm of in] {Preprocessing$(\bm{T}$)};
		\node (dec) [myblock, right=0.7cm of prepro] {Dec.};
		\node (inv) [myblock, right=0.7cm of dec] {$\bm{T}^{-1}$};
		\node (out) [right=1cm of inv] {$\hat{\bm{x}}$};
		\draw[->] (in) -- (prepro);
		\draw[->](prepro) -- (dec);
		\draw[->] (dec) -- (inv);
		\draw[->] (inv) -- (out);
	\end{tikzpicture}%
    \vspace*{-2mm}%
    \caption{Path of a GAED if an automorphism  $\bm{T}\in \mathrm{Aut}(\mathcal{C}$) is used.} \label{figure:aut_preprocessing}
    \vspace*{-2mm}%
 \end{figure}
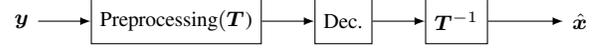

Let $\bm{x}\in\mathcal{C}$ be an arbitrary codeword 
that is transmitted over a binary memoryless channel.
     Consider an automorphism $\bm{T}\in \mathrm{GAut}(\mathcal{C})$, and define 
     \begin{equation}\label{equation:aut_sum_according_row}
         \tilde{\bm{x}}:=\bm{T} \bm{x}=\begin{pmatrix}
            \sum\limits_{i=1}^{n} T_{1,i} x_i, & \dots &, \sum\limits_{i=1}^{n} T_{n,i} x_i
         \end{pmatrix}^{\top}.
     \end{equation} 
     Then, to mimic the effect of this automorphism at the receiver the bit\--wise LLRs $L( y_i|x_i)$ are processed according to
      \begin{equation}\label{equation:preprocessing}
        L( y_j|\tilde{x}_j)= \mathboxplus_{
         \begin{subarray}{l}
 \,\,\,\, i=1,\\
 T_{j,i}=1 \end{subarray}}^{n}  L(y_i|x_i),\quad \forall j \in \{1,\dots,n\},
      \end{equation}
which follows immediately from (\ref{eq:boxplus}).
Note that permuting the bit-wise LLR vector $\pi_i(\bm{L})$ with $\pi_i\in\mathrm{Aut}(\mathcal{C})$ in Sec.~\ref{sec:aed} is a special case of (\ref{equation:preprocessing}). Hence, the proposed GAED algorithm naturally generalizes AED proposed in \cite{AED_RMcodes}.

According to (\ref{equation:preprocessing}), the numbers of non\--zero entries per row 
of $\bm{T}$ indicate the number of LLR values that participate in the boxplus summations. To quantify those effects, we conveniently define the \emph{weight} $\Omega(\bm{T})$ as the number of non-zero elements of $\bm{T}$ and ${\Delta(\bm{T})=\Omega(\bm{T})-n}$ as the \emph{weight over permutation}. Note that, due to $\bm{T}\in\mathrm{GL}_n(\mathbb{F})$, $\Omega(\bm{T})$ is lower bounded by $n$ which is only obtained if $\bm{T}\in\mathrm{Aut}(\mathcal{C})$.

Since the magnitude of an LLR indicates the reliability of the message, it is important to understand the influence of the weight $\Omega(\bm{T})$ of an automorphism 
on the magnitude of the LLR after the proposed preprocessing. It can be shown that 
the magnitude of the outgoing LLR is decreasing if more finite LLRs participate in the boxplus summation.
Thus, the preprocessing with an automorphism $\bm{T}\in \mathrm{GAut}(\mathcal{C})\setminus \mathrm{Aut}(\mathcal{C})$ leads to an information loss.
Hence, we expect the paths of GAED with such automorphisms to individually possess
a higher decoding error probability.
Still, as long as the performance degradation per path is not too severe, the decoding performance of the full GAED algorithm can improve.
Note that the complexity of GAED is comparable to AED because the preprocessing step can be interpreted as one check node update.

\vspace*{-1mm}

\section{Construction of Linear Codes with Sparse
Automorphisms}
\label{sec:code-construction}
Next, we propose a method to construct linear codes $\mathcal{C}$ with a specific, potentially sparse, automorphism ${\bm{T}\in \text{\Aut{}}}$. 
We observe that if $\bm{T}$ is sufficiently sparse, then often %
$\bm{T}^{-1}$ and $\bm{T}^2$ are also sparse automorphisms usable in GAED.
A possible approach is to find a non-singular CCM ${\bm{A}}$
such that ${\bm{A}^{-1}  \bm{T}  \bm{A} \in  \mathcal{Z}}$.
The following construction method
designs a code $\mathcal{C}$ along with an automorphism of weight $\Omega_{\text{obj}}$ and is based on the observation that the Frobenius normal form
is, in some cases, an element of $\mathcal{Z}$:
 \begin{enumerate}
            \item Choose $\Omega_{\text{obj}}$ close to $n$, i.e., $\Delta(\bm{T})$ close to zero.
     \item Sample a matrix $\bm{T}\in \mathrm{GL}_n(\mathbb{F})$ with $\Omega(\bm{T})=\Omega_{\text{obj}}$
     \begin{itemize}
         \item Determine the sizes $d_i$ of the Frobenius normal form $\bm{F}$ by solving a set of linear equations\cite{bhattacharya1994basic}.
         \item Evaluate if the matrices $\bm{B}_{f_i}$ can be ordered such that $\bm{F}$ is an element of $\mathcal{Z}$, e.g., by using Theorem \ref{theorem:frob_in:z}.
         \item Otherwise, repeat step 2).
     \end{itemize}
     \item Calculate $\bm{S}_{\text{F}}=\bm{A}^{-1}$, such that 
     $\bm{F}=\bm{A}^{-1} \bm{T} \bm{A}\in \mathcal{Z}$.
     \item Extract the PCM, denoted $\bm{H}_{\mathrm{c}}$, according to Theorem~\ref{lemma:codeinfo_in_a}.%
     \item Find an optimized PCM $\bm{H}_{\mathrm{opt}}$ based on $\bm{H}_{\mathrm{c}}$.
 \end{enumerate}
 Theorem \ref{theorem:frob_in:z} states a sufficient condition for the existence of a CCM for a given $\bm{T}$. A proof is given in the appendix. %
\begin{theorem}\label{theorem:frob_in:z}\vspace{-1mm}
    Let $d_1,\dots d_{j}$ be the sizes of the block matrices of the Frobenius normal form $\bm{F}$ of $\bm{T}$. If there exists a subset $\mathcal{J}\subseteq \{1,\dots,j\}$ with
     ${\sum_{i\in \mathcal{J}} d_i=k}$,
     then there exists an ordering of the $\bm{B}_{f_i}$ yielding an upper-right all-zero block of size ${(n-k)\times k}$ or ${k\times (n-k)}$  within $\bm{F}$. Hence,
     $\mathcal{C}(n,k)$ and $\mathcal{C}(n,n-k)$ can be constructed.\vspace{-1mm}
\end{theorem}

 Note that the method only yields a PCM. A generator matrix must still be determined, which may contain zero columns.
 Then, a suitable reduction of $\bm{G}$, $\bm{H}$ and $\bm{T}$ can be performed resulting in a code with smaller block length.

 It is not guaranteed that a code constructed with the proposed method has good properties as, e.g., large minimum Hamming distance.
 In addition, structural properties of the resulting PCM
 are not yet considered in the first four steps, but is subject to ongoing research. 
 Thus, we currently perform a heuristic optimization in which we randomly choose low-weight dual codewords to construct a full rank PCM.
 Note that the full dual codebook can be determined for all of the upcoming codes.

\vspace*{-1mm}

\vspace*{-1.5mm}%
\section{Results}%
\vspace*{-1mm}%
\label{sec:results}

\begin{table}
        \centering
        \caption{Code parameters}\label{table:constructed_codes}
        \vspace*{-2mm}
        \begin{tabular}{cccccccc}
                \toprule
                \textbf{Code} &     $\mathcal{C}_1$ & $\mathcal{C}_{1,\mathrm{ref}}$ &$\mathcal{C}_2$ & $\mathcal{C}_{2,\mathrm{ref}}$& $\mathcal{C}_3$  &   $\mathcal{C}_{\mathrm{BCH}}$  \\
                \midrule
                \vspace*{-2mm}\\
                                    $n$& $39$ & $39$ & $32$& $32$  & $63$    & $63$\\
                                    $k$& $24$ & $24$ & $16$& $16$  & $45$    & $45$\\
                        $d_\text{min}$ & $6$ &  $6$  & $5$ & $8$  &  $5$     & $7$
                \\
                \bottomrule
                \vspace*{-2mm}\\
        \end{tabular}
        \vspace*{-4mm}
\end{table}

We analyze the performance of three constructed binary codes $\mathcal{C}_1$\--$\mathcal{C}_3$, 
 two reference codes $\mathcal{C}_{1,\text{ref}}, \mathcal{C}_{2,\text{ref}}$ from \cite{Grassl:codetables} and a BCH code $\mathcal{C}_{\mathrm{BCH}}$, with parameters outlined in Table \ref{table:constructed_codes}. The minimum Hamming distances were obtained using the methods proposed in \cite{search_dmin}.
To evaluate the frame error rate (FER), we perform Monte\--Carlo simulations using an AWGN channel, accumulating at least 300 frame errors for each SNR.
The notation {GAED-$\ell$\--BP--$p$} denotes GAED consisting of $\ell$ BP path decoders performing $p$ iterations of normalized min-sum decoding (normalization constant $\frac{3}{4}$).
All GAEDs rely on three different automorphisms, namely the identity mapping $\bm{I}$, an automorphism $\bm{T}$ constructed according to Sec. \ref{sec:code-construction}, and its inverse $\bm{T}^{-1}$.
Additionally, as reference, we show results of a redundant row BP decoder, named {R-$\ell$\--BP--$p$} decoder, which performs BP decoding with $p$ iterations
using an overcomplete PCM consisting of $\ell\cdot (n-k)$ low-weight dual codewords. This approach is known to potentially improve BP decoding of short block codes \cite{MBBP2}.
Ordered statistics decoders are used to approximate the ML performances of all codes \cite{FosLin}. 

Fig. \ref{figure:39_24_perm}-\ref{figure:63_45} depict the FER over $E_{\mathrm{b}}/N_0$ for codes $\mathcal{C}_1$-$\mathcal{C}_3$ based on automorphisms with different weights over permutation.
Code $\mathcal{C}_{1}$ was constructed based on a permutation ${(\Delta(\bm{T})=0)}$, hence GAED equals AED. Codes $\mathcal{C}_2$ and $\mathcal{C}_3$ were designed to have automorphisms with ${\Delta(\bm{T})=10}$ and ${\Delta(\bm{T})=5}$, respectively, to show validity of the general approach.
For all constructed codes, the performance of GAED-3\--BP\--10 is compared against two decoders with comparable complexity, namely the BP\--30 and the {R-$3$\--BP-$10$}.
Additionally, Fig. \ref{figure:39_24_perm}-\ref{figure:63_45}
depict the performance of the ML decoder for the constructed codes and some reference codes. Both are intended as an indication
of the best achievable performance. Since the paper at hand is intended as a proof of concept, a natural gap to this performance is still observed.
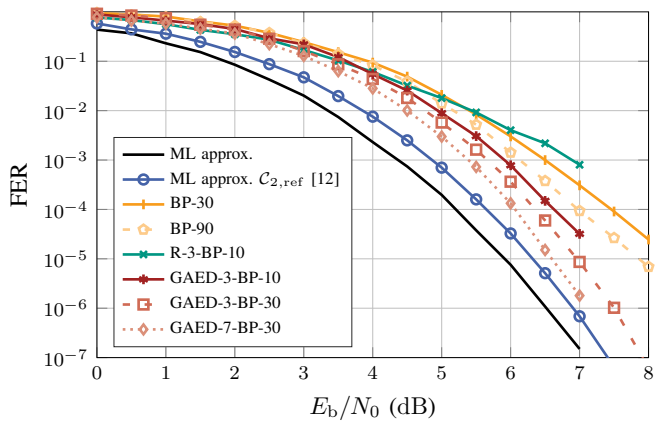
\begin{figure}[t]
                \centering
        \definecolor{mycolor1}{rgb}{0.46600,0.67400,0.18800}%
\definecolor{mycolor2}{rgb}{0.49400,0.18400,0.55600}%
\definecolor{mycolor3}{rgb}{1.00000,0.44700,0.74100}%
\definecolor{mycolor4}{rgb}{0.00000,0.44700,0.74100}%
\definecolor{kit-green100}{rgb}{0,.59,.51}

\definecolor{kit-green100}{rgb}{0,.59,.51}
\definecolor{kit-green70}{rgb}{.3,.71,.65}
\definecolor{kit-green50}{rgb}{.50,.79,.75}
\definecolor{kit-green30}{rgb}{.69,.87,.85}
\definecolor{kit-green15}{rgb}{.85,.93,.93}
\definecolor{KITgreen}{rgb}{0,.59,.51}

\definecolor{KITpalegreen}{RGB}{130,190,60}
\colorlet{kit-maigreen100}{KITpalegreen}
\colorlet{kit-maigreen70}{KITpalegreen!70}
\colorlet{kit-maigreen50}{KITpalegreen!50}
\colorlet{kit-maigreen30}{KITpalegreen!30}
\colorlet{kit-maigreen15}{KITpalegreen!15}

\definecolor{KITblue}{rgb}{.27,.39,.66}
\definecolor{kit-blue100}{rgb}{.27,.39,.67}
\definecolor{kit-blue70}{rgb}{.49,.57,.76}
\definecolor{kit-blue50}{rgb}{.64,.69,.83}
\definecolor{kit-blue30}{rgb}{.78,.82,.9}
\definecolor{kit-blue15}{rgb}{.89,.91,.95}

\definecolor{KITyellow}{rgb}{.98,.89,0}
\definecolor{kit-yellow100}{cmyk}{0,.05,1,0}
\definecolor{kit-yellow70}{cmyk}{0,.035,.7,0}
\definecolor{kit-yellow50}{cmyk}{0,.025,.5,0}
\definecolor{kit-yellow30}{cmyk}{0,.015,.3,0}
\definecolor{kit-yellow15}{cmyk}{0,.0075,.15,0}

\definecolor{KITorange}{rgb}{.87,.60,.10}
\definecolor{kit-orange100}{cmyk}{0,.45,1,0}
\definecolor{kit-orange70}{cmyk}{0,.315,.7,0}
\definecolor{kit-orange50}{cmyk}{0,.225,.5,0}
\definecolor{kit-orange30}{cmyk}{0,.135,.3,0}
\definecolor{kit-orange15}{cmyk}{0,.0675,.15,0}

\definecolor{KITred}{rgb}{.63,.13,.13}
\definecolor{kit-red100}{cmyk}{.25,1,1,0}
\definecolor{kit-red70}{cmyk}{.175,.7,.7,0}
\definecolor{kit-red50}{cmyk}{.125,.5,.5,0}
\definecolor{kit-red30}{cmyk}{.075,.3,.3,0}
\definecolor{kit-red15}{cmyk}{.0375,.15,.15,0}

\definecolor{KITpurple}{RGB}{160,0,120}
\colorlet{kit-purple100}{KITpurple}
\colorlet{kit-purple70}{KITpurple!70}
\colorlet{kit-purple50}{KITpurple!50}
\colorlet{kit-purple30}{KITpurple!30}
\colorlet{kit-purple15}{KITpurple!15}

\definecolor{KITcyanblue}{RGB}{80,170,230}
\colorlet{kit-cyanblue100}{KITcyanblue}
\colorlet{kit-cyanblue70}{KITcyanblue!70}
\colorlet{kit-cyanblue50}{KITcyanblue!50}
\colorlet{kit-cyanblue30}{KITcyanblue!30}
\colorlet{kit-cyanblue15}{KITcyanblue!15}
\definecolor{mycolor1}{rgb}{0.46600,0.67400,0.18800}%
\definecolor{mycolor2}{rgb}{0.49400,0.18400,0.55600}%
\definecolor{mycolor3}{rgb}{1.00000,0.44700,0.74100}%
\definecolor{mycolor4}{rgb}{0.00000,0.44700,0.74100}%
\definecolor{kit-green100}{rgb}{0,.59,.51}
\definecolor{mycolor1}{rgb}{0.46600,0.67400,0.18800}%
\definecolor{mycolor2}{rgb}{0.49400,0.18400,0.55600}%
\definecolor{mycolor3}{rgb}{1.00000,0.44700,0.74100}%
\definecolor{mycolor4}{rgb}{0.00000,0.44700,0.74100}%
\definecolor{kit-green100}{rgb}{0,.59,.51}

\begin{tikzpicture}[scale=0.92,spy using outlines={rectangle, magnification=1.3}]

\begin{axis}[
width=.9\columnwidth,
height=5cm,
at={(0.758in,0.645in)},
scale only axis,
xmin=0,
xmax=8,
xlabel style={font=\color{white!15!black}},
xlabel={$E_{\mathrm{b}}/N_0$ ($\si{dB}$)},
ymode=log,
ymin=1e-07,
ymax=1,
ytick={1e-1,1e-2,1e-3,1e-4,1e-5,1e-6,1e-7},
ylabel style={font=\color{white!15!black}},
ylabel={FER},
axis background/.style={fill=white},
xmajorgrids,
ymajorgrids,
legend style={at={(0.03,0.03)}, anchor=south west, legend cell align=left, align=left, draw=white!15!black, font=\scriptsize}
]
\addplot [color=black, line width=1.1pt]
  table[row sep=crcr]{%
0	0.43731778425656\\
0.5	0.367647058823529\\
1	0.230769230769231\\
1.5	0.153374233128834\\
2	0.0850099178237461\\
2.5	0.0424989375265618\\
3	0.0201369311316955\\
3.5	0.0074279488957116\\
4	0.00230671638921995\\
4.5	0.000740528638710888\\
5	0.000196864473902661\\
5.5	3.77576630121025e-05\\
6	7.58388439733335e-06\\
6.5	1.08e-06\\
7	1.5e-07\\
};
\addlegendentry{ML approx.}

\addplot [color=KITblue, line width=1.1pt,mark=o, mark options={solid, KITblue}]
  table[row sep=crcr]{%
0	5.81818182e-01\\
0.5	4.41860465e-01\\
1	3.61474435e-01\\
1.5	2.47254982e-01\\
2	1.54080081e-01\\
2.5	8.73814314e-02\\
3	4.72673560e-02\\
3.5	1.97057108e-02\\
4	7.55213832e-03\\
4.5	2.50137823e-03\\
5	7.01123181e-04\\
5.5	1.59767160e-04\\
6	3.26202028e-05\\
6.5	5.11925067e-06\\
7	6.86723648e-07\\
7.5 6.5e-08\\
};
\addlegendentry{ML approx. $\mathcal{C}_{\mathrm{2,ref}}$ \cite{Grassl:codetables}}

\addplot [color=kit-orange100, line width=1.1pt, mark=|, mark options={solid, kit-orange100}]
  table[row sep=crcr]{%
0	0.945121951219512\\
0.5	0.882022471910112\\
1	0.802992518703242\\
1.5	0.641025641025641\\
2	0.525547445255474\\
2.5	0.376190476190476\\
3	0.243562231759657\\
3.5	0.152488425925926\\
4	0.0938369199238765\\
4.5	0.0488538595752361\\
5	0.0208919244676577\\
5.5	0.00836220736328304\\
6	0.00300809769698972\\
6.5	0.000998153684593395\\
7	0.000307386603368223\\
7.5	9.076e-05\\
8	2.457e-05\\
};
\addlegendentry{BP\--30}

\addplot [color=kit-orange50,loosely dashed,mark=pentagon, line width=1.1pt, mark options={solid, kit-orange50}]
  table[row sep=crcr]{%
0	   0.952380952380952\\     
0.5	0.845070422535211\\
1	  0.783289817232376\\
1.5	0.623700623700624\\
2	  0.519930675909879\\
2.5	0.373599003735990\\
3	  0.231124807395994\\
3.5	0.139082058414465\\
4	  0.0825536598789213\\
4.5	0.0385505011565150\\
5	  0.0140056022408964\\
5.5	0.00520300386756621\\
6	  0.00141724025529221\\
6.5 0.000377728141501998\\
7 9.33328894835922e-05\\
7.5 2.64619779746373e-05\\
8 6.88819120186033e-06\\
};
\addlegendentry{ BP-90 }

\addplot [color=KITgreen, line width=1.1pt, mark=x,mark options={solid, KITgreen}]
  table[row sep=crcr]{%
0	7.712082262210796513e-01\\
0.5	6.818181818181817677e-01\\
1	5.628517823639774820e-01\\
1.5	4.322766570605187098e-01\\
2	3.525264394829611958e-01\\
2.5	2.724795640326975743e-01\\
3	1.690140845070422504e-01\\
3.5	1.035554021401449837e-01\\
4	6.087662337662337608e-02\\
4.5	3.212679374598415016e-02\\
5	1.794043774668101915e-02\\
5.5	9.202171712524156305e-03\\
6	4.022741900879639801e-03\\
6.5	2.180898239288154954e-03\\
7	8.020682667037400040e-04\\
};
\addlegendentry{R-3-BP\--10}

\addplot [color=KITred, line width=1.1pt, mark=asterisk, mark options={solid, KITred}]
  table[row sep=crcr]{%
0	8.869179600886918369e-01\\
0.5	7.766990291262135804e-01\\
1	6.655574043261230921e-01\\
1.5	5.649717514124293904e-01\\
2	4.494382022471909988e-01\\
2.5	2.915451895043731922e-01\\
3	2.194185408667032322e-01\\
3.5	1.213960546282245867e-01\\
4	5.501306560308073107e-02\\
4.5	2.551671344730798568e-02\\
5	8.889086424142758278e-03\\
5.5	3.050547573289405338e-03\\
6	7.825476229137769569e-04\\
6.5 1.489999999999999897e-04\\
7 3.199999999999999855e-05\\
};
\addlegendentry{GAED-3\--BP\--10}

\addplot [color=kit-red70, line width=1.1pt,loosely dashed, mark=square, mark options={solid, kit-red70}]
  table[row sep=crcr]{%
0	0.914634146341463\\
0.5	0.842696629213483\\
1	0.748129675810474\\
1.5	0.549450549450549\\
2	0.437956204379562\\
2.5	0.285714285714286\\
3	0.160944206008584\\
3.5	0.0868055555555556\\
4	0.043917435221783\\
4.5	0.0180494555080922\\
5	0.00573953968891695\\
5.5	0.00162163038719128\\
6	0.000364323499837269\\
6.5	5.97577540167668e-05\\
7	8.65552665763722e-06\\
7.5	1.02e-06\\
8	6e-08\\
};
\addlegendentry{GAED-3\--BP-30 }

\addplot [color=kit-red50,dotted,mark=diamond ,line width=1.1pt, mark options={solid, kit-red50}]
  table[row sep=crcr]{%
0   7.766990291262135804e-01\\
0.5 6.644518272425249394e-01\\
1	5.788712011577423766e-01\\
1.5	4.581901489117983783e-01\\
2	3.633060853769300436e-01\\
2.5	2.216066481994459769e-01\\
3	1.282873636946760820e-01\\
3.5	6.299212598425196763e-02\\
4	2.822267692090594735e-02\\
4.5	1.007835924311522094e-02\\
5	2.984584620435451040e-03\\
5.5 7.264340261697858252e-04\\
6   1.330000000000000073e-04\\
6.5 1.509999999999999942e-05\\
7    1.799999999999999919e-06\\
};
\addlegendentry{GAED-7-BP-30 }

\end{axis}

\end{tikzpicture}%
%
        \vspace*{-2mm}%
        \caption{ Performance of different decoders for code $\mathcal{C}_1(39,24)$. The GAED-7 relies on $\bm{T}^\alpha$ with ${\Delta(\bm{T}^\alpha)=0}$ and $\alpha\in\{-3,\dots,3\}$ in its paths.
        } \label{figure:39_24_perm} 
\end{figure}

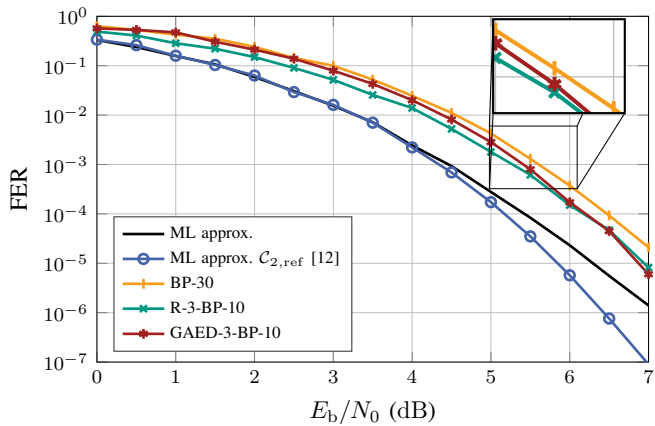
\begin{figure}[t]
                \centering
        \definecolor{mycolor1}{rgb}{0.46600,0.67400,0.18800}%
\definecolor{mycolor2}{rgb}{0.49400,0.18400,0.55600}%
\definecolor{mycolor3}{rgb}{1.00000,0.44700,0.74100}%
\definecolor{mycolor4}{rgb}{0.00000,0.44700,0.74100}%
\definecolor{kit-green100}{rgb}{0,.59,.51}

\definecolor{kit-green100}{rgb}{0,.59,.51}
\definecolor{kit-green70}{rgb}{.3,.71,.65}
\definecolor{kit-green50}{rgb}{.50,.79,.75}
\definecolor{kit-green30}{rgb}{.69,.87,.85}
\definecolor{kit-green15}{rgb}{.85,.93,.93}
\definecolor{KITgreen}{rgb}{0,.59,.51}

\definecolor{KITpalegreen}{RGB}{130,190,60}
\colorlet{kit-maigreen100}{KITpalegreen}
\colorlet{kit-maigreen70}{KITpalegreen!70}
\colorlet{kit-maigreen50}{KITpalegreen!50}
\colorlet{kit-maigreen30}{KITpalegreen!30}
\colorlet{kit-maigreen15}{KITpalegreen!15}

\definecolor{KITblue}{rgb}{.27,.39,.66}
\definecolor{kit-blue100}{rgb}{.27,.39,.67}
\definecolor{kit-blue70}{rgb}{.49,.57,.76}
\definecolor{kit-blue50}{rgb}{.64,.69,.83}
\definecolor{kit-blue30}{rgb}{.78,.82,.9}
\definecolor{kit-blue15}{rgb}{.89,.91,.95}

\definecolor{KITyellow}{rgb}{.98,.89,0}
\definecolor{kit-yellow100}{cmyk}{0,.05,1,0}
\definecolor{kit-yellow70}{cmyk}{0,.035,.7,0}
\definecolor{kit-yellow50}{cmyk}{0,.025,.5,0}
\definecolor{kit-yellow30}{cmyk}{0,.015,.3,0}
\definecolor{kit-yellow15}{cmyk}{0,.0075,.15,0}

\definecolor{KITorange}{rgb}{.87,.60,.10}
\definecolor{kit-orange100}{cmyk}{0,.45,1,0}
\definecolor{kit-orange70}{cmyk}{0,.315,.7,0}
\definecolor{kit-orange50}{cmyk}{0,.225,.5,0}
\definecolor{kit-orange30}{cmyk}{0,.135,.3,0}
\definecolor{kit-orange15}{cmyk}{0,.0675,.15,0}

\definecolor{KITred}{rgb}{.63,.13,.13}
\definecolor{kit-red100}{cmyk}{.25,1,1,0}
\definecolor{kit-red70}{cmyk}{.175,.7,.7,0}
\definecolor{kit-red50}{cmyk}{.125,.5,.5,0}
\definecolor{kit-red30}{cmyk}{.075,.3,.3,0}
\definecolor{kit-red15}{cmyk}{.0375,.15,.15,0}

\definecolor{KITpurple}{RGB}{160,0,120}
\colorlet{kit-purple100}{KITpurple}
\colorlet{kit-purple70}{KITpurple!70}
\colorlet{kit-purple50}{KITpurple!50}
\colorlet{kit-purple30}{KITpurple!30}
\colorlet{kit-purple15}{KITpurple!15}

\definecolor{KITcyanblue}{RGB}{80,170,230}
\colorlet{kit-cyanblue100}{KITcyanblue}
\colorlet{kit-cyanblue70}{KITcyanblue!70}
\colorlet{kit-cyanblue50}{KITcyanblue!50}
\colorlet{kit-cyanblue30}{KITcyanblue!30}
\colorlet{kit-cyanblue15}{KITcyanblue!15}
\definecolor{mycolor1}{rgb}{0.46600,0.67400,0.18800}%
\definecolor{mycolor2}{rgb}{0.49400,0.18400,0.55600}%
\definecolor{mycolor3}{rgb}{1.00000,0.44700,0.74100}%
\definecolor{mycolor4}{rgb}{0.00000,0.44700,0.74100}%
\definecolor{kit-green100}{rgb}{0,.59,.51}
\begin{tikzpicture}[scale=0.92,spy using outlines={rectangle, magnification=1.5}]

\begin{axis}[%
width=.9\columnwidth,
height=5cm,
at={(0.758in,0.645in)},
scale only axis,
xmin=0,
xmax=7,
xlabel style={font=\color{white!15!black}},
xlabel={$E_{\mathrm{b}}/N_0$ ($\si{dB}$)},
ymode=log,
ymin=1e-07,
ymax=1,
yminorticks=true,
ytick={1,1e-1,1e-2,1e-3,1e-4,1e-5,1e-6,1e-7},
ylabel style={font=\color{white!15!black}},
ylabel={FER},
axis background/.style={fill=white},
xmajorgrids,
ymajorgrids,
legend style={at={(0.03,0.03)}, anchor=south west, legend cell align=left, align=left, draw=white!15!black,font=\scriptsize}
]
\addplot [color=black, line width=1.1pt]
  table[row sep=crcr]{%
0	0.331858407079646\\
0.5	0.239043824701195\\
1	0.157397691500525\\
1.5	0.106496272630458\\
2	0.0591016548463357\\
2.5	0.03060599877576\\
3	0.0154966682163335\\
3.5	0.00715563506261181\\
4	0.00242977937603266\\
4.5	0.000928789694149554\\
5	0.00027793527443329\\
5.5	8.35736773003028e-05\\
6	2.31466103987999e-05\\
6.5	5.57120341522198e-06\\
7	1.4e-06\\
7.5	2.3e-07\\
8	0\\
};
\addlegendentry{ML approx.}

\addplot [color=KITblue,mark=o, line width=1.1pt, mark options={solid, KITblue}]
  table[row sep=crcr]{%
0	0.337078651685393\\
0.5	0.260642919200695\\
1	0.159235668789809\\
1.5	0.103305785123967\\
2	0.0632644453816955\\
2.5	0.0292625829106516\\
3	0.0161134386077989\\
3.5	0.00701787218115467\\
4	0.00221829501844882\\
4.5	0.00069072062883206\\
5	0.000172598691932047\\
5.5	3.50509711222059e-05\\
6	5.71614710954257e-06\\
6.5	7.6e-07\\
7	9e-08\\
};
\addlegendentry{ML approx. $\mathcal{C}_{2,\mathrm{ref}}$ \cite{Grassl:codetables}}

\addplot [color=kit-orange100, line width=1.1pt, mark=|, mark options={solid, kit-orange100}]
  table[row sep=crcr]{%
0	6.289308176100628645e-01\\
0.5	5.309734513274336765e-01\\
1	4.261363636363636465e-01\\
1.5	3.521126760563380031e-01\\
2	2.413515687851970903e-01\\
2.5	1.431297709923664008e-01\\
3	9.970089730807576989e-02\\
3.5	5.231037489102004917e-02\\
4	2.461033634126333011e-02\\
4.5	1.110165414646782298e-02\\
5	4.254594962559564604e-03\\
5.5 1.304932208771754408e-03\\
6   3.745033149785097345e-04\\
6.5 9.299999999999999727e-05\\
7   2.099999999999999884e-05\\
};
\addlegendentry{BP-30}

\addplot [color=KITgreen, line width=1.1pt,mark=x, mark options={solid, KITgreen}]
  table[row sep=crcr]{%
0	4.909983633387888791e-01\\
0.5	4.081632653061224580e-01\\
1	2.849002849002849058e-01\\
1.5	2.212389380530973559e-01\\
2	1.493280238924838133e-01\\
2.5	9.033423667570009141e-02\\
3	5.179558011049723715e-02\\
3.5	2.562788313685289768e-02\\
4	1.389725297632834512e-02\\
4.5	5.253940455341506166e-03\\
5	1.799618480882052935e-03\\
5.5 6.166837625109718212e-04\\
6   1.510000000000000112e-04\\
6.5 4.671822386033618139e-05\\
7   8.199999999999999417e-06\\
};
\addlegendentry{R-3-BP\--10}

\addplot [color=KITred, line width=1.1pt, mark=asterisk, mark options={solid, KITred}]
  table[row sep=crcr]{%
0	5.692599620493358215e-01\\
0.5	5.309734513274336765e-01\\
1	4.731861198738170349e-01\\
1.5	3.054989816700611205e-01\\
2	2.112676056338028130e-01\\
2.5	1.383125864453665366e-01\\
3	7.932310946589106460e-02\\
3.5	4.277159965782720286e-02\\
4	2.005481649842904071e-02\\
4.5	8.141333550435560576e-03\\
5	2.841258867095380891e-03\\
5.5 7.984966968853305768e-04\\
6   1.720000000000000066e-04\\
6.5 4.500000000000000283e-05\\
7   6.000000000000000152e-06\\
};
\addlegendentry{GAED-3\--BP\--10}

\coordinate (spypoint) at (axis cs:5,0.0005);
			\coordinate (spyviewer) at (axis cs:5.3,0.025);	
			\spy[width=1.75cm,height=1.25cm, thin, spy connection path={\draw(tikzspyonnode.south west) -- (tikzspyinnode.south west);\draw (tikzspyonnode.south east) -- (tikzspyinnode.south east);
			\draw (tikzspyonnode.north west) -- (tikzspyinnode.north west);\draw (tikzspyonnode.north east) -- (intersection of  tikzspyinnode.north east--tikzspyonnode.north east and tikzspyinnode.south east--tikzspyinnode.south west);
			;}] on (spypoint) in node at (spyviewer);
		\coordinate (a) at ($(axis cs:-10.8/1.4,-0.12)+(spyviewer)$);
		\coordinate[label={[font=\small,text=black]right:$10^{-2}$}] (b) at ($(axis cs:+10.8/1.4,-0.12)+(spyviewer)$);

\end{axis}
\end{tikzpicture}%
%
        \vspace*{-2mm}%
        \caption{ Performance of different decoders for code $\mathcal{C}_2(32,16)$ with ${\Delta(\bm{T})=10}$ and $\Delta(\bm{T}^{-1})=13$.} \label{figure:32_1610over} 
\end{figure}
\begin{figure}[t]
                \centering
        \definecolor{mycolor1}{rgb}{0.46600,0.67400,0.18800}%
\definecolor{mycolor2}{rgb}{0.49400,0.18400,0.55600}%
\definecolor{mycolor3}{rgb}{1.00000,0.44700,0.74100}%
\definecolor{mycolor4}{rgb}{0.00000,0.44700,0.74100}%
\definecolor{kit-green100}{rgb}{0,.59,.51}

\definecolor{kit-green100}{rgb}{0,.59,.51}
\definecolor{kit-green70}{rgb}{.3,.71,.65}
\definecolor{kit-green50}{rgb}{.50,.79,.75}
\definecolor{kit-green30}{rgb}{.69,.87,.85}
\definecolor{kit-green15}{rgb}{.85,.93,.93}
\definecolor{KITgreen}{rgb}{0,.59,.51}

\definecolor{KITpalegreen}{RGB}{130,190,60}
\colorlet{kit-maigreen100}{KITpalegreen}
\colorlet{kit-maigreen70}{KITpalegreen!70}
\colorlet{kit-maigreen50}{KITpalegreen!50}
\colorlet{kit-maigreen30}{KITpalegreen!30}
\colorlet{kit-maigreen15}{KITpalegreen!15}

\definecolor{KITblue}{rgb}{.27,.39,.66}
\definecolor{kit-blue100}{rgb}{.27,.39,.67}
\definecolor{kit-blue70}{rgb}{.49,.57,.76}
\definecolor{kit-blue50}{rgb}{.64,.69,.83}
\definecolor{kit-blue30}{rgb}{.78,.82,.9}
\definecolor{kit-blue15}{rgb}{.89,.91,.95}

\definecolor{KITyellow}{rgb}{.98,.89,0}
\definecolor{kit-yellow100}{cmyk}{0,.05,1,0}
\definecolor{kit-yellow70}{cmyk}{0,.035,.7,0}
\definecolor{kit-yellow50}{cmyk}{0,.025,.5,0}
\definecolor{kit-yellow30}{cmyk}{0,.015,.3,0}
\definecolor{kit-yellow15}{cmyk}{0,.0075,.15,0}

\definecolor{KITorange}{rgb}{.87,.60,.10}
\definecolor{kit-orange100}{cmyk}{0,.45,1,0}
\definecolor{kit-orange70}{cmyk}{0,.315,.7,0}
\definecolor{kit-orange50}{cmyk}{0,.225,.5,0}
\definecolor{kit-orange30}{cmyk}{0,.135,.3,0}
\definecolor{kit-orange15}{cmyk}{0,.0675,.15,0}

\definecolor{KITred}{rgb}{.63,.13,.13}
\definecolor{kit-red100}{cmyk}{.25,1,1,0}
\definecolor{kit-red70}{cmyk}{.175,.7,.7,0}
\definecolor{kit-red50}{cmyk}{.125,.5,.5,0}
\definecolor{kit-red30}{cmyk}{.075,.3,.3,0}
\definecolor{kit-red15}{cmyk}{.0375,.15,.15,0}

\definecolor{KITpurple}{RGB}{160,0,120}
\colorlet{kit-purple100}{KITpurple}
\colorlet{kit-purple70}{KITpurple!70}
\colorlet{kit-purple50}{KITpurple!50}
\colorlet{kit-purple30}{KITpurple!30}
\colorlet{kit-purple15}{KITpurple!15}

\definecolor{KITcyanblue}{RGB}{80,170,230}
\colorlet{kit-cyanblue100}{KITcyanblue}
\colorlet{kit-cyanblue70}{KITcyanblue!70}
\colorlet{kit-cyanblue50}{KITcyanblue!50}
\colorlet{kit-cyanblue30}{KITcyanblue!30}
\colorlet{kit-cyanblue15}{KITcyanblue!15}
\definecolor{mycolor1}{rgb}{0.46600,0.67400,0.18800}%
\definecolor{mycolor2}{rgb}{0.49400,0.18400,0.55600}%
\definecolor{mycolor3}{rgb}{1.00000,0.44700,0.74100}%
\definecolor{mycolor4}{rgb}{0.00000,0.44700,0.74100}%
\definecolor{kit-green100}{rgb}{0,.59,.51}
\begin{tikzpicture}[scale=0.92,spy using outlines={rectangle, magnification=2}]

\begin{axis}[%
width=.9\columnwidth,
height=5cm,
at={(0.758in,0.645in)},
scale only axis,
xmin=1,
xmax=6,
xlabel style={font=\color{white!15!black}},
xlabel={$E_{\mathrm{b}}/N_0$ ($\si{dB}$)},
ymode=log,
ymin=1e-06,
ymax=1,
ytick={1e-1,1e-2,1e-3,1e-4,1e-5,1e-6},
yminorticks=true,
ylabel style={font=\color{white!15!black}},
ylabel={FER},
axis background/.style={fill=white},
xmajorgrids,
ymajorgrids,
legend style={at={(0.03,0.03)}, anchor=south west, legend cell align=left, align=left, draw=white!15!black,font=\scriptsize}
]

\addplot [color=KITblue, line width=1.1pt,mark=o, mark options={solid, KITblue}]
  table[row sep=crcr]{%
0.00  6.329e-01\\
0.50  4.975e-01\\
1.00  3.704e-01\\
1.50  2.445e-01\\
2.00  1.447e-01\\
2.50  7.353e-02\\
3.00  2.595e-02\\
3.50  7.918e-03\\
4.00  2.134e-03\\
4.50  4.751e-04\\
5.00  5.337e-05\\
5.50  6.300e-06\\
6   7.11946581e-07\\
};
\addlegendentry{ML approx. $\mathcal{C}_{\mathrm{BCH}}$}

\addplot [color=kit-blue70, line width=1.1pt,mark=triangle, mark options={solid, kit-blue70}]
  table[row sep=crcr]{%
0.0 9.259259259259259300e-01\\
0.5 9.287925696594426794e-01\\
1.0 8.219178082191780366e-01\\
1.5 6.564551422319474527e-01\\
2.0 5.244755244755244794e-01\\
2.5 3.640776699029126262e-01\\
3.0 2.427184466019417508e-01\\
3.5 1.220504475183075699e-01\\
4.0 6.384337092998509933e-02\\
4.5 2.570253598355037861e-02\\
5.0 7.212231945379363705e-03\\
5.5 2.290006412017953855e-03\\
6.0 4.538124786329958112e-04\\
6.5 8.600000000000000331e-05\\
7.0 2.000000000000000164e-05\\
};
\addlegendentry{BP\--30 $\mathcal{C}_{\mathrm{BCH}}$}

\addplot [color=black, line width=1.1pt]
  table[row sep=crcr]{%
0	6.91358025e-01\\
0.5 5.67951318e-01\\
1   4.28790199e-01\\
1.5 2.59981430e-01\\
2   1.61290323e-01\\
2.5 6.74698795e-02\\
3   2.75834893e-02\\
3.5 9.06618314e-03\\
4   2.19662974e-03\\
4.5 5.14540359e-04\\
5   9.50628433e-05\\
5.5 1.72278718e-05\\
6   2.31674393e-06\\ 
};
\addlegendentry{ML approx.}

\addplot [color=kit-orange100, line width=1.1pt, mark=|, mark options={solid,kit-orange100}]
  table[row sep=crcr]{%
0.0 9.404388714733542542e-01\\
0.5 9.036144578313253239e-01\\
1.0 8.174386920980926119e-01\\
1.5 6.696428571428570953e-01\\
2.0 5.825242718446601575e-01\\
2.5 3.989361702127659504e-01\\
3.0 2.500000000000000000e-01\\
3.5 1.462701121404192950e-01\\
4.0 6.096321885795569911e-02\\
4.5 2.329192546583850817e-02\\
5.0 8.324315325064514035e-03\\
5.5 2.211736950751990761e-03\\
6.0 5.413432892379149881e-04\\
};
\addlegendentry{BP\--30}

\addplot [color=KITgreen, line width=1.1pt,mark=x, mark options={solid, KITgreen}]
  table[row sep=crcr]{%
0.0 9.090909090909090606e-01\\
0.5 8.620689655172413257e-01\\
1.0 7.537688442211055717e-01\\
1.5 6.160164271047228191e-01\\
2.0 4.862236628849270770e-01\\
2.5 3.401360544217686965e-01\\
3.0 2.139800285306704686e-01\\
3.5 1.088929219600725945e-01\\
4.0 5.937067088858104247e-02\\
4.5 2.377744313228184053e-02\\
5.0 7.464543418760885909e-03\\
5.5 2.822918332972627117e-03\\
6.0 5.933943342708963538e-04\\
};
\addlegendentry{R-3-BP\--10}

\addplot [color=KITred, line width=1.1pt, mark=asterisk, mark options={solid, KITred}]
  table[row sep=crcr]{%
0.0 9.523809523809523281e-01\\
0.5 8.645533141210374195e-01\\
1.0 7.874015748031496509e-01\\
1.5 6.651884700665188221e-01\\
2.0 5.802707930367504430e-01\\
2.5 4.291845493562231884e-01\\
3.0 2.475247524752475226e-01\\
3.5 1.280409731113956451e-01\\
4.0 6.006006006006005954e-02\\
4.5 2.109704641350210880e-02\\
5.0 6.363887062217602093e-03\\
5.5 1.632173402102239310e-03\\
6.0 2.709999999999999735e-04\\
};
\addlegendentry{GAED-3\--BP\--10}

\addplot [color=kit-red70, line width=1.1pt,loosely dashed, mark=square, mark options={solid, kit-red70}]
  table[row sep=crcr]{%
0.0 9.345794392523364413e-01\\
0.5 8.522727272727272929e-01\\
1.0 7.731958762886598224e-01\\
1.5 6.289308176100628645e-01\\
2.0 5.084745762711864181e-01\\
2.5 3.083247687564234507e-01\\
3.0 2.036659877800407470e-01\\
3.5 1.139817629179331326e-01\\
4.0 4.068348250610252098e-02\\
4.5 1.467136150234741795e-02\\
5.0 3.835581410215431735e-03\\
5.5 8.125413041829626301e-04\\
6.0 1.190000000000000058e-04\\
};
\addlegendentry{GAED-3\--BP-30 }

\coordinate (spypoint) at (axis cs:5.12,0.0006);
			\coordinate (spyviewer) at (axis cs:4.8,0.04);	
			\spy[width=1.75cm,height=1.25cm, thin, spy connection path={\draw(tikzspyonnode.south west) -- (tikzspyinnode.south west);
   \draw (tikzspyonnode.south east) -- (tikzspyinnode.south east);
			\draw (tikzspyonnode.north east) -- (tikzspyinnode.north east);
   \draw (tikzspyonnode.north west) -- (intersection of  tikzspyinnode.north west--tikzspyonnode.north west and tikzspyinnode.south east--tikzspyinnode.south west);
			;}] on (spypoint) in node at (spyviewer);
		\coordinate (a) at ($(axis cs:-10.8/1.4,-0.12)+(spyviewer)$);
		\coordinate[label={[font=\small,text=black]right:$10^{-2}$}] (b) at ($(axis cs:+10.8/1.4,-0.12)+(spyviewer)$);

\end{axis}
\end{tikzpicture}%
%
        \vspace*{-2mm}%
        \caption{ Performance of different decoders for code $\mathcal{C}_3(63,45)$ with $\Delta(\bm{T})=5$ and $\Delta(\bm{T}^{-1})=6$.
        } \label{figure:63_45} 
\end{figure}
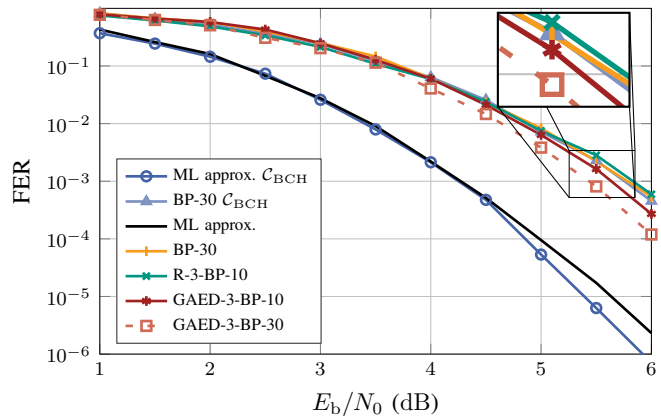

We observe that GAED-3\--BP\--10 results in lower error rates compared to BP\--30. 
For code $\mathcal{C}_{1}$, GAED-3\--BP-10 shows a gain of $0.6\,\si{dB}$ compared to BP\--30 at an FER of $10^{-3}$ and also outperforms R-3-BP-10 by $1\,\si{dB}$. 
When decoding $\mathcal{C}_2$ and $\mathcal{C}_3$, GAED-3\--BP-10 also is able to yield a gain over BP-30. 
For code $\mathcal{C}_{2}$, R-3-BP-10 obtains a gain compared to GAED-3\--BP-10. However, in higher SNR regime, GAED-3-BP-10 is able to close the gap to R-3-BP-10.

Comparing GAED-3-BP-30 and BP-90, i.e., when decoding with increased complexity, similar improvements can be observed. 
For code $\mathcal{C}_3$, the performance of GAED-3-BP-10 and BP-90 coincide. Hence, the latter was omitted for clarity. 
As additional reference, we simulated the performance of BP-30 for $\mathcal{C}_{\mathrm{BCH}}$ 
using the optimized PCM \emph{1min} from \cite{channelcodes}. It can be observed that the performance of BP-30 for codes $\mathcal{C}_3$ and $\mathcal{C}_{\mathrm{BCH}}$ coincide over the whole SNR regime. 
Our simulations indicate that increasing the number
of iterations does not yield further improvement for all decoders of $\mathcal{C}_3$.
Note that GAED for $\mathcal{C}_2$ and $\mathcal{C}_3$ relies on elements from ${\mathrm{GAut}(\mathcal{C})\setminus \mathrm{Aut}(\mathcal{C})}$, serving as proof that the generalized automorphism group can, in fact, improve decoding.

\vspace*{-1mm}

\section{Conclusion}
\vspace{-1mm}
\label{sec:conclusion}
In this paper, we have shown that the application of the more general definition of automorphisms prevailing in linear algebra
in GAED can be used to improve decoding compared to BP decoding. 
One important advantage is that this more general definition is expected to simplify the search for suitable transformation significantly. 
To this end, we first analyzed generalized automorphisms of linear codes based on non\--singular mappings of $\mathbb{F}^n$ and proved a specification of their structure introducing the CCM. Then, we described the resulting effects at the receiver and reasoned that generalized automorphisms should possess sparse matrices to prevent severe information loss.
Additionally, we introduced a method to construct linear codes together with potentially sparse automorphisms. 
Finally, we discussed the decoding performance of GAED for three exemplary constructed codes with varying code sizes and rates.
In all cases, GAED improved decoding when compared to equal complexity BP decoding. 
Therefore, this approach is very promising to enable alternative code designs and to improve decoding performance.

\newpage
\appendix

\begin{proof}[Proof of Theorem \ref{lemma:codeinfo_in_a}]
Consider a linear code $\mathcal{C}$ with PCM $\bm{H}$.
From Theorem \ref{theorem:core}, it follows that there exists at least one CCM $\bm{A}\in \mathrm{GL}_n( \mathbb{F})$ to transform $\bm{H}$ into $\tilde{\bm{H}}$. If $\tilde{\bm{H}}$ 
is multiplied from the right with some matrix $\bm{Z}\in \mathcal{Z}$, 
then $\tilde{\bm{H}} \bm{Z} =\tilde{\bm{H}}$ holds. Thus, $\bm{A}_1=\bm{A} \bm{Z}$ also transforms $\bm{H}$ into the desired form. Therefore, the matrix $\bm{A}$ is not unique.

If $\bm{A}^{-1}$ has the structure described in  (\ref{equation:ainv_code_info}) and the code rate $r$ is known, then 
 the PCM $\bm{H}$ of $\mathcal{C}$ can be extracted from the inverse CCM.
 Thus, $\bm{A}^{-1}$ characterizes $\mathcal{C}$ because a linear code is fully defined by its PCM. Consequently, $\bm{A}$ also must characterize the code. 
To prove  (\ref{equation:ainv_code_info}), $\tilde{\bm{H}}$ is multiplied with $\bm{A}^{-1}$ from the right. Assuming that $$\bm{A}^{-1}=\begin{pmatrix}
    \bm{U}_{(n-k)\times n}\\
    \bm{\Lambda}_{k\times n}
\end{pmatrix},$$
it can be seen that
  $  \tilde{\bm{H}} \bm{A}^{-1} = \bm{H} \bm{A} \bm{A}^{-1}=\bm{H}$ and
\begin{align*}
        &\,\tilde{\bm{H}} \bm{A}^{-1} =
    \begin{bmatrix}
     \bm{I}_{(n-k)\times (n-k)}&\bm{0}_{(n-k)\times k}
    \end{bmatrix}  \begin{pmatrix}
        \bm{U}\\
        \bm{\Lambda}
    \end{pmatrix} \overset{!}{=}\bm{H}\\
   \Longleftrightarrow &\, \bm{U} + \bm{0}_{(n-k)\times k} =\bm{U}\overset{!}{=}\bm{H}.
\end{align*}
Therefore, the PCM is contained within $\bm{A}^{-1}$.
\end{proof}

\begin{proof}[Proof of Theorem \ref{theorem:frob_in:z}]
 Let $d_1,\dots d_{j}$ be the sizes of the block matrices of the Frobenius normal form $\bm{F}$ of $\bm{T}$ and let there exist a subset $$\mathcal{J}\subseteq \{1,\dots,j\}=:\mathcal{I}$$ with
     ${\sum_{i\in \mathcal{J}} d_i=k}$. Then, because the sizes of the block matrices of the Frobenius normal form necessarily sum up to $n$, i.e., ${\sum_{i\in \mathcal{I}} d_i=n}$,
     it follows that ${\sum_{i\in \mathcal{I}\setminus \mathcal{J}} d_i=n-k}$.
      Without loss of generality, assume that $\mathcal{J}=\{1,\dots,m\}$.
     Define the matrices
     \begin{align*}
      \bm{F}_{\mathcal{J}}&:=\begin{pmatrix}
      \bm{B}_{f_1}& 0& \cdots&0 \\
      0 & \ddots & \ddots &\vdots \\
      \vdots & \ddots & \ddots  & 0 \\
      0 & \cdots & 0 & \bm{B}_{f_{m}} 
      \end{pmatrix}\in  \mathbb{F}^{k\times k},\text{ and }\\
      \bm{F}_{\mathcal{I}\setminus\mathcal{J}}&:=\begin{pmatrix}
      \bm{B}_{f_{m+1}}& 0& \cdots&0 \\
      0 & \ddots & \ddots &\vdots \\
      \vdots & \ddots & \ddots  & 0 \\
      0 & \cdots & 0 & \bm{B}_{f_{j}} 
      \end{pmatrix} \in  \mathbb{F}^{(n-k)\times (n-k)}.
     \end{align*}
    Then, $\bm{F}_{\mathcal{J}}$ and $\bm{F}_{\mathcal{I}\setminus\mathcal{J}}$ are of the form
    $$
   \begin{pmatrix}
    \setlength{\arraycolsep}{1pt}
      *&*&0&\cdots&0 \\
      \vdots&\ddots&\ddots&\ddots &\vdots \\
      \vdots&\ddots&\ddots &\ddots &0\\
      *&\cdots&\cdots &* & *\\
     * & *&\cdots &\cdots&*
    \end{pmatrix},
    $$
and there exists
\begin{align*}
\bm{F}&=\begin{pmatrix}
\bm{F}_{\mathcal{I}\setminus\mathcal{J}}&\bm{0}\\
    \bm{0}&\bm{F}_{\mathcal{J}}
\end{pmatrix}\\
&=\left(\begin{smallmatrix}
\setlength{\arraycolsep}{0.1pt}
        *&*&0&\cdots&0&                      &&&&\\
        \vdots&\ddots&\ddots&\ddots &\vdots& &&&& \\
        \vdots&\ddots&\ddots &\ddots&\vdots&         &&\bm{0}_{(n-k)\times k} && \\
        *&\cdots&\cdots &* & *                &&&&\\
        * & *&\cdots &\cdots&*                &&&&\\
        &&&&    &*&*&0&\cdots&0 \\
        &&&&    &\vdots&\ddots&\ddots&\ddots &\vdots \\
        &&\bm{0}_{k\times(n-k)} &&    &\vdots&\ddots&\ddots &\ddots &0\\
        &&&&    &*&\cdots&\cdots &* & *\\
        &&&&    &* & *&\cdots &\cdots&*
\end{smallmatrix}\right)\\
&\in \mathcal{Z}(n,k)
\end{align*}
Similarly, there exists $\bm{F}\in \mathcal{Z}(n,n-k)$.
Therefore, according to Theorem \ref{lemma:codeinfo_in_a}, $\mathcal{C}(n,k)$ as well as $\mathcal{C}(n,n-k)$ can be constructed.
\end{proof}
\end{document}